\documentclass[letterpaper, amsfonts, amssymb, amsmath, pra, aps, reprint, showkeys, nofootinbib]{revtex4-2}

\usepackage{array}
\usepackage{amsmath}
\usepackage{amsfonts}
\usepackage{amsthm}
\usepackage{algorithm}
\usepackage{algpseudocode}
\usepackage{bm, bbm}
\usepackage{booktabs} 
\usepackage{enumerate}
\usepackage{graphicx}
\usepackage{hyperref}
\usepackage{longtable}
\usepackage{multirow}
\usepackage{natbib}
\usepackage{physics}
\usepackage{setspace}
\usepackage{threeparttable}
\usepackage{upgreek}

\renewcommand{\mu}{\upmu}
\hypersetup{
    colorlinks = true,
    allcolors = black
}

\graphicspath{{figures/}}

\newcommand{\beginsupplement}{%
    \setcounter{table}{0}
    \renewcommand{\thetable}{S\arabic{table}}%
    \setcounter{figure}{0}
    \renewcommand{\thefigure}{S\arabic{figure}}%
    \renewcommand{\thesubsection}{\arabic{subsection}}
 }
\newcommand{\ri}{\mathrm{i}}

\DeclareMathOperator*{\argmax}{argmax} 

\newtheorem{theorem}{\bf Theorem}
\newtheorem{remark}{\bf Remark}

\makeatletter
\newenvironment{breakablealgorithm}
  {
    \begin{center}
      \refstepcounter{algorithm}
      \hrule height.8pt depth0pt \kern2pt
      \parskip 0pt
      \renewcommand{\caption}[2][\relax]{
        {\raggedright\textbf{\fname@algorithm~\thealgorithm} ##2\par}%
        \ifx\relax##1\relax 
          \addcontentsline{loa}{algorithm}{\protect\numberline{\thealgorithm}##2}%
        \else 
          \addcontentsline{loa}{algorithm}{\protect\numberline{\thealgorithm}##1}%
        \fi
        \kern2pt\hrule\kern2pt
     }
  }
  {
     \kern2pt\hrule\relax
   \end{center}
  }
\makeatother

\begin{document}
\title[Zhu et al]{A Realizable GAS-based Quantum Algorithm for Traveling Salesman Problem}

\author{Jieao Zhu$^1$}
\author{Yihuai Gao$^1$}
\author{Hansen Wang$^2$}
\author{Tiefu Li$^{3, \dag, *}$}
\author{Hao Wu$^{4, \ddag, *}$}

\affiliation{
$^1$Department of Electronic Engineering, Tsinghua University, Beijing, China \\
$^2$Institute for Interdisciplinary Information Sciences, Tsinghua University, Beijing, China\\
$^3$School of Integrated Circuits, Tsinghua University, Beijing, China\\
$^4$Department of Mathematical Sciences, Tsinghua University, Beijing, China \\
}

\let\thefootnote\relax\footnotetext{$^\dag$ litf@tsinghua.edu.cn}
\let\thefootnote\relax\footnotetext{$^\ddag$ hwu@tsinghua.edu.cn}
\let\thefootnote\relax\footnotetext{$^*$ Corresponding authors }
\date{\today}

\begin{abstract}
    The paper proposes a quantum algorithm for the traveling salesman problem (TSP) based on the Grover Adaptive Search (GAS), which can be successfully executed on IBM's Qiskit library. Under the GAS framework, there are at least two fundamental difficulties that limit the application of quantum algorithms for combinatorial optimization problems. One difficulty is that the solutions given by the quantum algorithms may not be feasible. The other difficulty is that the number of qubits of current quantum computers is still very limited, and it cannot meet the minimum requirements for the number of qubits required by the algorithm. In response to the above difficulties, we designed and improved the Hamiltonian Cycle Detection (HCD) oracle based on mathematical theorems. It can automatically eliminate infeasible solutions during the execution of the algorithm. On the other hand, we design an anchor register strategy to save the usage of qubits. The strategy fully considers the reversibility requirement of quantum computing, overcoming the difficulty that the used qubits cannot be simply overwritten or released. As a result, we successfully implemented the numerical solution to TSP on IBM's Qiskit. For the seven-node TSP, we only need 31 qubits, and the success rate in obtaining the optimal solution is 86.71\%.
\end{abstract}


\maketitle

\section{Introduction}

Quantum computing is widely acknowledged as a revolutionary paradigm in computational technology, which enables significant speedup over classical computing for a wide range of problems. 
The power of quantum computing originates from quantum algorithms~\cite{montanaro2016quantum}, which can achieve exponential speedups on certain computing tasks thanks to quantum superposition~\cite{Frank2019Quantum}. 
According to different computational models, quantum algorithms can be divided into two broad categories, and they are intrinsically equivalent in polynomial time~\cite{aharonov2008adiabatic}. 
The algorithms based on the quantum gate model belong to the first category, such as Shor's integer factorization~\cite{shor1994algorithms}, Grover's unstructured database search~\cite{grover1996fast,boyer1998tight}, and variational quantum algorithms~\cite{cerezo2021variational, biamonte2017quantum}.
Another category is the adiabatic quantum algorithm~\cite{farhi2000quantum, aharonov2008adiabatic, albash2018adiabatic}, which is based on the evolution of the ground state over a time-varying Hamiltonian. 
It uses quantum effects instead of thermal effects, enabling tunneling from one state to another~\cite{APOLLONI1989233,finnila1994quantum,das2008colloquium}. It is often understood as the quantum extension of the simulated annealing algorithm.

Due to the powerful quantum superposition characteristics of quantum computing, it has potential advantages in solving some NP-hard combinatorial optimization problems. 
Thus, the following problems have been initially explored: the Quadratic Unconstrained Binary Optimization problem~\cite{gilliam2021grover, harwood2021formulating, mcgeoch2013experimental, ushijima2021multilevel, das2008colloquium}, the Max-Cut problem~\cite{guerreschi2019qaoa,fuchs2021efficient}, the Quadratic Assignment Problem~\cite{khumalo2021investigation,ajagekar2019quantum}, and the Traveling Salesman Problem (TSP)~\cite{martovnak2004quantum,salehi2022unconstrained}. 
For the above problems, Grover Adaptive Search (GAS) algorithm~\cite{durr1996quantum} is considered to be the most competitive algorithm framework since it ensures quadratic speedup in most cases and provides exact solutions with guaranteed success probability~\cite{nielsen2002quantum}. 
Moreover, the quantum computer technology related to GAS and the quantum gate model has also developed rapidly in recent years~\cite{schmidt2003realization,han2000genetic,gilliam2021grover,ishikawa2021quantum}. 
Apart from GAS, other different quantum algorithms include quantum annealers~\cite{martovnak2004quantum, das2008colloquium}, which exploit the adiabatic evolution to find the solution, and Variational Quantum Algorithms (VQA)~\cite{guerreschi2019qaoa,harwood2021formulating}, which provide a hybrid quantum-classical approach to approximately solve combinatorial optimization problems. 
In this work, we focus on developing the GAS algorithm for TSP. 
We endeavor to discuss two fundamental difficulties in designing GAS-type algorithms for combinatorial optimization problems. 
One is the feasibility of the solution given by the quantum algorithms. The other is that the quantum algorithms require a large number of qubits, which currently cannot be achieved by quantum computers in the NISQ~\cite{preskill2018quantum} era.
Therefore, without loss of generality, we are devoted to addressing these two difficulties for TSP, and we hope that our efforts will provide insightful approaches and lead to more follow-up works on quantum combinatorial solvers.

TSP is one of the most well-known NP-hard combinatorial optimization problems. 
This problem considers a number of cities connected by paths of various  lengths, and the salesman tries to determine the shortest cyclic tour that visits each city exactly once. 
There have been extensive studies on this problem, e.g., the exact methods~\cite{laporte1992traveling, chauhan2012survey}, the approximation methods~\cite{christofides1976worst}, the heuristics methods~\cite{helsgaun2000effective, johnson1990local}, and the AI-based methods~\cite{bengio2021machine, lombardi2018boosting}. 
However, these methods always fail to overcome the hurdle of NP-hard problems. 
Thus, researchers hope to solve it by designing quantum algorithms. 
As we discussed earlier, there are two ways. 
One is the quantum annealing algorithm~\cite{heim2017designing, warren2013adapting, kieu2019travelling}, which exploits the quantum fluctuation property to find the TSP solution. 
Another approach is based on the gate model and GAS that belong to our interests. 
This paper is inspired by Srinivasan et al.'s work~\cite{srinivasan2018efficient}, which was unsuccessful since illegal tours could not be excluded. 
Later, IBM's Qiskit group\footnote{\url{https://qiskit.org/textbook/ch-paper-implementations/tsp.html}} tackled this problem with a correctly-designed phase estimation process. 
However, they mentioned that ``this process is imcomplete'', since all the legal Hamiltonian cycles need to be manually imported as candidates before the quantum algorithm is executed. 
In fact, finding all valid Hamiltonian cycles is inherently an NP-hard problem~\cite{akiyama1980np}. 
Moreover, there is no successful numerical experiments of Srinivasan et al.'s work, and IBM's Qiskit group only present the results of TSP with four nodes, which is obviously too trivial.

This work focuses on the two difficulties mentioned above. 
First, we present and prove the Cycle Determination Theorem, which converts the valid cycle determinations into problems that can be tested using quantum circuits. 
Based on the theorem, the Hamiltonian Cycle Detection (HCD) oracle is designed here.
Moreover, we also improved the Cycle Determination Theorem to determine the valid cycles with fewer tests by quantum circuits. 
Thus, the improved Hamiltonian Cycle Detection oracle is designed with fewer qubits and circuit layers. 
Secondly, we design the anchor register strategy for the problem that qubits can neither be overwritten nor be released due to the reversibility of quantum computation.
It significantly saves the use of the number of qubits.
This makes it possible to simulate some non-trivial problems on quantum computers or quantum simulators, such as seven-node TSP.

\section{Results}
    The classical $N$-city TSP is described as follows: each city $v_i\, (0\le i \le N-1)$ has $d_i\,(2\le d_i \le N-1)$ roads connected to other cities. And there is at most one road connecting any two cities.  
    The solution is the Hamiltonian cycle with the smallest cost. 
    The TSP graph is complete when $d_i=N-1, \forall i$, and it is called $d$-sparse when $d=\max_i\{d_i\}<N-1$.
    Note that in order to simplify the problem, a complete TSP problem can be converted into several sparse TSP problems by heuristically pruning the unlikely edges in the tour~\cite{wong2009efficient}.
    It is well-known that complete TSP and sparse TSP are both NP-hard~\cite{csaba2002approximability}. 
    The interest of this work is to design a quantum algorithm to give exact solutions to the above problems, and require lower complexity than classical algorithms. 
    
    As discussed earlier, we would like to establish a general approach to overcoming the difficulties of Grover's adaptive search for combinatorial optimization problems, starting with solving TSP. 
    Our main work is divided into three parts. 
    (1) Encode all the candidate TSP solutions into quantum states. 
    (2) Construct the Cycle Length Comparing (CLC) oracle that selects TSP cycles with a cost less than the threshold. 
    (3) Construct the Hamiltonian Cycle Detection (HCD) oracle that excludes illegal TSP tours. 
    From this, we can extract the desired solution from a uniform superposition state by the standard Grover's procedure~\cite{grover1996fast,gilliam2021grover}. 
    After several repetitions of quantum measurements, we can obtain the optimal solution with high probability.
    
\subsection{TSP Problem Encoding}
    For $N$-city TSP, the traveling salesman has at most $d\leq N-1$ choices of roads in each city, so we only need $m=\lceil \log_{2}{d} \rceil$ qubits to encode his choice.
    Accordingly, each candidate solution can be specified by a quantum eigenstate on $mN$ qubits. In the following, we refer it as the cycle register $\ket{C}$. 
    The undirected graph in TSP is represented by the adjacency matrix $A=(a_{i,j}), 0 \leq i,j < N$. 
    This is a symmetric matrix whose element $a_{i,j}> 0$ represents the traveling cost between two cities $i$ and $j$, where the self-loops $(i\to i)$ and the non-existing edges $(i\to j)$ are excluded by infinite edge costs: $a_{i,i}=\infty$, and $a_{i,j}=\infty$, respectively.   
    To store the connections, we use the adjacency list 
    \begin{equation}
        P_i = (0,1,\cdots, i-1, i+1,\cdots,N-1),\;\forall\, 0\leq i\leq N-1,
        \label{eqn:adj_list_def}
    \end{equation}
    where the self-loop edge $(i\to i)$ and the non-existing edges should be excluded from the adjacency list $P_i$, as is mentioned above. 
    Note that both the complete TSP $\#P_{i} = N-1,\; \forall i$, and the sparse TSP $\#P_i<N-1,\; \forall i$ can be processed by this encoding. 
    Thus, the cycle register $\ket{C}$ can be represented as follows
    \begin{equation}
        \ket{C} = \otimes_{i=0}^{N-1}\ket{C_i}=\otimes_{i=0}^{N-1}\ket{c_{i,(m-1)} \cdots c_{i,1}c_{i,0}}.
        \label{eq:cycle_register_form}
    \end{equation}
    Here $C_i$ encodes the path choice of the salesman at city $v_i$, and $c_{i,m-1}\cdots c_{i,1}c_{i,0}$ is the binary representation of $C_i$. 
    According to the adjacency lists \eqref{eqn:adj_list_def} and the integer value $C_i$, we can obtain the binary code of the next city $P_i[C_i]$ from the $i$-th city.
    
    Next, we present a simple example to demonstrate the qubit encoding scheme. For a $6$-city complete TSP, $N=6,\; d=5$, only $m=3$ qubits are required for encoding the choice. 
    The adjacency lists are given by 
    \begin{equation}
        \begin{aligned}
            P_{0} &= \{1, 2, 3, 4, 5\}, \quad P_{1} = \{0, 2, 3, 4, 5\}, \\
            P_{2} &= \{0, 1, 3, 4, 5\}, \quad P_{3} = \{0, 1, 2, 4, 5\}, \\
            P_{4} &= \{0, 1, 2, 3, 5\}, \quad P_{5} = \{0, 1, 2, 3, 4\}.  
        \end{aligned}
        \label{eqn:adjacency_lists}
    \end{equation}
    \begin{table}
        \centering 
        \caption{A simple example of the qubit encoding scheme.}
        \label{tab:cycle_register_example}
        \begin{tabular}{c|c|c|c|c|c|c}
            \toprule[1.5pt]
                City index $i$ & 0 & 1 & 2 & 3 & 4 & 5 \\ \hline 
                $\ket{C_i}$ & $\ket{010}$ & $\ket{011}$ & $\ket{001}$ & $\ket{100}$ & $\ket{000}$ & $\ket{010}$ \\ \hline 
                $C_i$ & 2 & 3 & 1 & 4 & 0 & 2 \\ \hline 
                Next city $P_i[C_i]$ & 3 & 4 & 1 & 5 & 0 & 2 \\ 
            \bottomrule[1.5pt]
        \end{tabular}
    \end{table}
    For example, in {\bf Table.~\ref{tab:cycle_register_example}}, the quantum state 
    \begin{equation*}
        \ket{C}=\ket{010}\otimes\ket{011}\otimes\ket{001}\otimes\ket{100}\otimes\ket{000}\otimes\ket{010},
    \end{equation*}
    represents the Hamiltonian cycle 
    \begin{equation*}
        0\rightarrow 3\rightarrow 5 \rightarrow 2 \rightarrow 1 \rightarrow 4 \rightarrow 0.
    \end{equation*}

\subsection{Framework}
    Next, we sketch the ingredients of our Grover Adaptive Search~\cite{gilliam2021grover}-based Quantum Algorithm for Traveling Salesman Problem (GQ-TSP). It consists of three major steps (see also {\bf Fig.~\ref{fig:top_design}} for quantum circuit):

    \begin{enumerate}
        \item Initialization. The cycle register $\ket{C}$, as a group of qubits, is initialized to the uniform superposition state. And the result qubit $\ket{R}$ is initialized to $\ket{-}$ for Grover's search.
        \item Grover iteration. For $k=1,2,\cdots,I_{\rm opt}$,  execute the following five sub-procedures. 
            \begin{itemize}
                \item Label the states whose length are smaller than the threshold $C_{\rm th}$ by the CLC oracle.
                \item Label the states that correspond to the legal cycles by the HCD oracle.
                \item Flip the result qubit $\ket{R}$ through a C$^2$NOT gate whenever the states are labeled by both of the two oracles.  
                \item Apply the CLC oracle and HCD oracle again to release the intermediate qubits for reuse.
                \item Complete the Grover iteration by the diffusion operator $2\ket{s}\bra{s}-I$, where $\ket{s}$ is the uniform superposition state.
            \end{itemize}
        \item Measurement. The cycle register $\ket{C}$ is measured to obtain the solution to TSP. 
    \end{enumerate}
    \begin{figure*}
        \centering
        \includegraphics[width=0.8\linewidth]{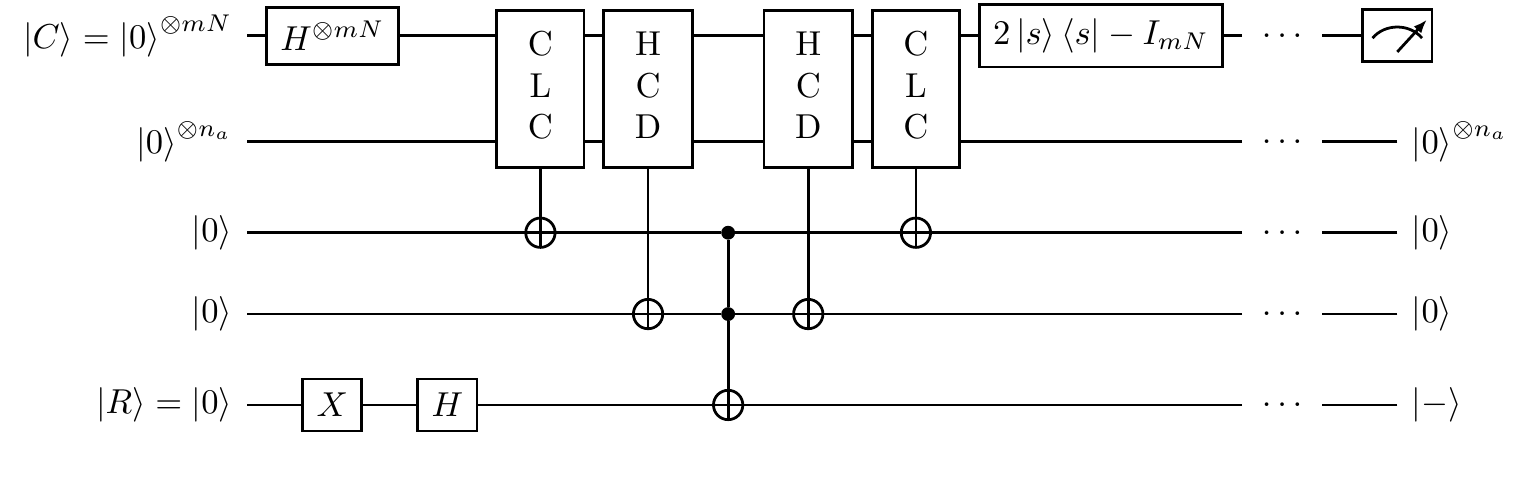}

        \caption{Illustration of the quantum circuit of GQ-TSP. }
        \label{fig:top_design}
    \end{figure*}

    Thus, the cycle register $\ket{C}$ contains the probability information of all the cycles that satisfy the requirements of the HCD and CLC oracles. 
    Moreover, the Hamiltonian cycles with smaller costs have larger probability amplitude. 
    Thus, a timely quantum measurement yields a valid optimal cycle with a high probability. 
    It should be mentioned that $I_{\rm opt}$ and $C_{\rm th}$ are very important to the effect of the algorithm and will be discussed in detail later.

\subsection{Cycle Length Comparing Oracle}
    The purpose of designing the Cycle Length Comparing (CLC) oracle is to filter all the candidate solutions according to their cycle costs, which is defined by
    \begin{equation}
        {\rm cost}(\ket{C}) = \sum_{j=0}^{N-1} a_{j, P_j[C_j]}.
        \label{eqn:cost_definition}
    \end{equation}
    Following Srinivasan, et.al.'s work~\cite{srinivasan2018efficient}, the quantum computation of \eqref{eqn:cost_definition} is realized by the controlled $U$-operator, which is composed of $N$ smaller operators $U_j, 0\leq j<N$ defined as 
    \begin{equation}
        U_j={\rm diag}\left( \exp\left(\ri \theta _{j, 0}\right),  \exp\left(\ri\theta _{j, 1}\right), \cdots, \exp\left(\ri\theta _{j, 2^{m}-1}\right)\right), 
        \label{eq:U_op}
    \end{equation}
    \begin{equation}
        \textrm{with}\;\theta_{j,k} = \left\{ \begin{aligned}  
            a_{j,P_j[k]}, \quad &0\leq k < \#P_j\\
            0,\quad & k\geq \#P_j.
        \end{aligned}\right.
    \end{equation}
    Corresponding to different path choices, $U_j$ introduces different phase shifts, i.e. 
    \begin{equation}
        U_j\ket{C_j}=\exp(\ri \theta_{j, C_j})\ket{C_j}.
    \end{equation}
    Thus, the effect of the $U$-operator on $\ket{C}$ is given as follows
    \begin{equation}
        \begin{aligned}
            U\ket{C}&=\bigotimes_{j=0}^{N-1}U_j\ket{C_j}=\bigotimes_{j=0}^{N-1}\exp(\ri \theta_{j,C_j})\ket{C_j}\\
            &=\prod_{j=0}^{N-1}\exp(\ri \theta_{j,C_j})\bigotimes_{j=0}^{N-1}\ket{C_j}=\exp(\ri\cdot{\rm cost}(\ket{C})) \ket{C}.
        \end{aligned}
    \end{equation} 
    The $U$-operator acts as the cost computation module in the quantum circuit of the CLC oracle. 
    Later, the detailed implementation of the $U$-operator will be discussed in the {\bf Methods} section, in which we will elaborate on the techniques to reduce the number of qubits. 

    We now turn to the construction of the CLC oracle.
    We can divide it into three steps (see also {\bf Fig.~2} for illustration).
    First, apply the standard phase estimation (QPE) routine~\cite{nielsen2002quantum} to obtain the cost values.
    Here we need to execute the $U$-operator for arbitrary $\ket{\ell}$ times, where $\ket{\ell}$ is prepared into a uniform superposition state by the $t$ Hadamard gates. 
    The next thing is to label the states with costs smaller than $C_{\rm th}$. 
    Finally, we have to free up the intermediate ancillary qubits for reuse by the mirrored gates.  

    \begin{figure*}
        \centering
        \includegraphics[width=0.7\linewidth]{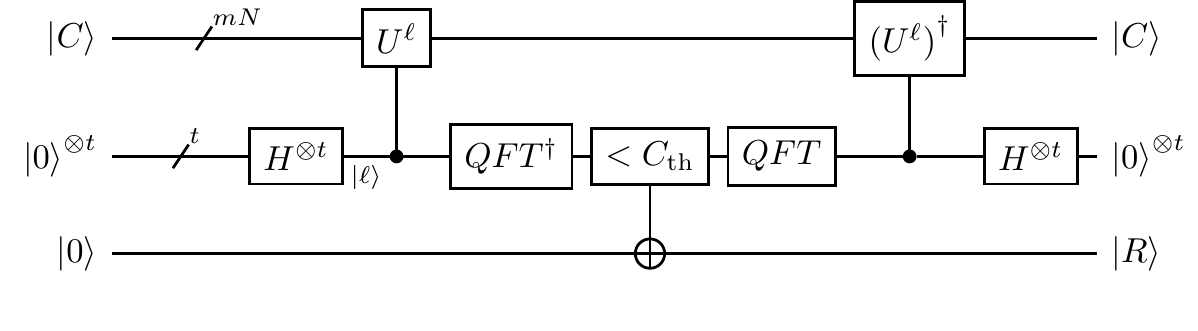}

        \caption{Quantum circuit of the CLC oracle.}
        \label{fig:CLC oracle}
    \end{figure*}

\subsection{Hamiltonian Cycle Detection Oracle} 
    The purpose of designing the Hamiltonian Cycle Detection (HCD) oracle is to pick out the legal Hamiltonian cycles. Its implementation relies on the following theorem.
    \begin{theorem}[Cycle Determination Theorem]  \label{thm_cycle_determination}
        Consider the function $\pi: \{0,1,\cdots,N-1\}\to \{0,1,\cdots,N-1\}$ that corresponds to the cycle register $\ket{C}$. Then $\ket{C}$ is a Hamiltonian cycle if and only if: 1) $\forall 1 \leqslant j \leqslant N-1$, $\pi^j(0) \neq 0$,\quad  2) $\pi^N(0) = 0$.
    \end{theorem}
    \begin{proof}
        The necessity is trivial. Here we only prove the sufficiency part. We first show that $\pi$ is a permutation. 
        In fact, we only need to prove the following:
            \begin{equation*}
                \pi^i(0)\neq \pi^{i+s}(0), \quad \forall 0\leq i < i+s < N.
                \label{eqn:unique_elements_1}
            \end{equation*}
            If not, there would exist $\pi^i(0)=\pi^{i+s}(0)$ for some $0<s< N-i$. Applying $\pi^{N-i}$ on both sides yields
            \begin{equation*}
                0=\pi^{N-i}(\pi^i(0))=\pi^{(N-i)+(i+s)}(0)=\pi^s(0),
            \end{equation*}
            which contradicts the assumption 1). 
            
            It remains to be shown that $\pi$ is a cycle. According to the permutation decomposition theorem~\cite[Theorem 1.1]{rotman2012introduction},
            \begin{equation*}
                \pi=\tau_1\tau_2\cdots\tau_k,
            \end{equation*} 
            where $\tau_i,\;1\le i\le k$ are disjoint cycles, and $k\geq 1$.
            We would like to prove $k=1$. If not, then there exists $\tau_j$ that satisfies 
            \begin{equation*}
                \tau_j^{\ell}(0)=0,\;\ell:={\rm ord}(\tau_j)<N.
            \end{equation*}
            Thus, we have 
            \begin{equation*}
                \pi^{\ell}(0)=0,
            \end{equation*}
            which leads to a contradiction. 
    \end{proof}
    Based on the above {\bf Theorem~\ref{thm_cycle_determination}}, we present the detailed design of the HCD oracle. The HCD oracle contains the following two steps: 
    \begin{enumerate}
        \item Traveling (see {\bf Fig.~\ref{fig:index_forwarder}} for illustration). According to the cycle register $\ket{C}$, the locations $\ket{I_j}$, $j=1,2,\cdots,N$ of the salesman at each step are calculated sequentially by the index forwarder $F$ gates.
        \item Checking (see {\bf Fig.~\ref{fig:Hamiltonian Cycle Detector}} for illustration). Perform checks on each of the locations, and detect whether the cycle $\ket{C}$ is valid according to {\bf Theorem~1}. 
        The checking part is described by the following quantum-logical expression:
        \begin{equation}
            R_{\rm HCD} = \left(\prod_{j=1}^{N-1}{\rm OR}(\ket{I_j})\right)\cdot{\rm NOR}(\ket{I_N}),
            \label{eqn:HCD_logical}
        \end{equation}
        where the product sumbol $\prod_j$ denote the logical AND operation, and the gate ${\rm OR}(\ket{I_j})$ represents bitwise logical OR operation on each of the $n$ qubits in the location register $\ket{I_j}$.
    \end{enumerate}
    After applying the HCD oracle, if the result qubit is initialized in superposition to $\ket{-}:=(\ket{0}-\ket{1})/\sqrt{2}$, then the state after applying the HCD oracle is described by 
    \begin{equation*}
        {\rm HCD}(\ket{C}\ket{-}) = (-1)^{R_{\rm HCD}}\ket{C}\ket{-},
    \end{equation*} 
    which introduces a cycle-specific negative sign to label the valid cycles.  

    In practice, all the checking procedures will be inverted by an ``uncomputation'' process to restore the ancillary qubits ${\rm OR}(\ket{I_1}),\cdots,{\rm OR}(\ket{I_6})$ to their initial state $\ket{0}$.     
    It is worth mentioning that the index forwarder $F$ plays an important role in the traveling part. And we will describe its detailed implementation in the {\bf Methods} section. 
    \begin{figure}
        \centering
        \includegraphics[width=\linewidth]{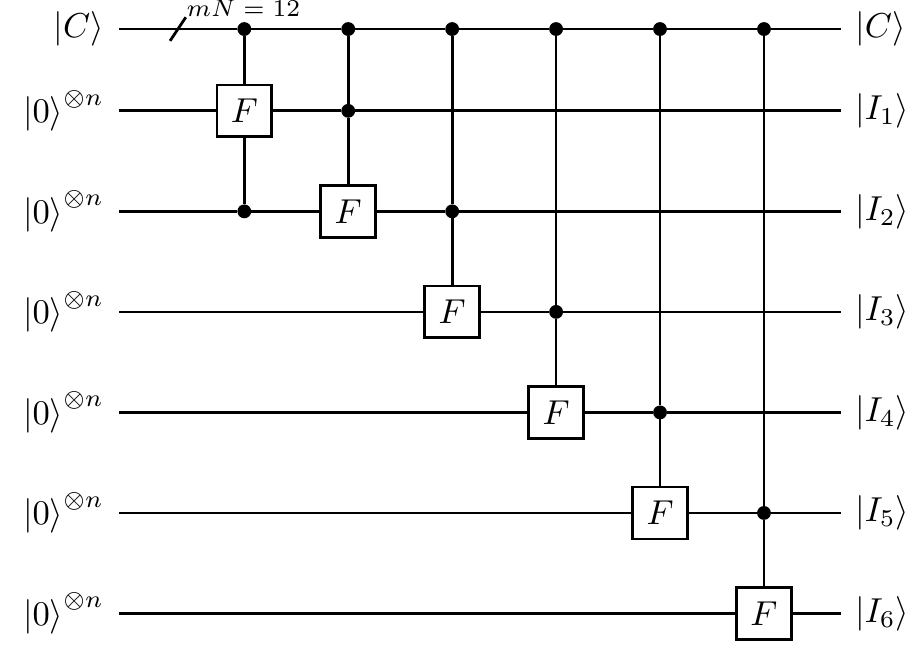}
        
        \caption{Quantum circuit illustration ($N=6$) of the traveling part of the HCD oracle.}
        \label{fig:index_forwarder}
    \end{figure}

    \begin{figure}
        \centering
        \includegraphics[width=\linewidth]{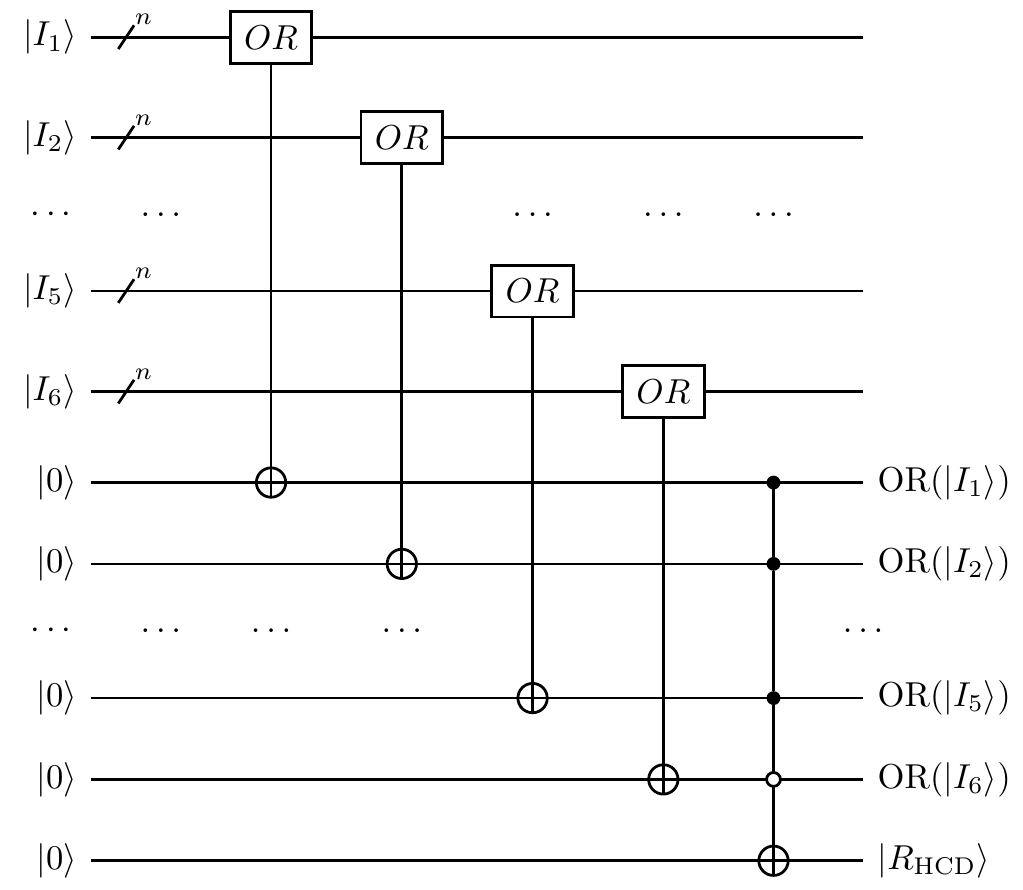}
            
        \caption{Quantum circuit illustration ($N=6$) of the checking part of the HCD oracle.}
        \label{fig:Hamiltonian Cycle Detector}
    \end{figure}

    Below,  we discuss the qubit consumption of the HCD oracle. 
    To record the salesman's locations $\ket{I_k}$, $nN$ qubits are required, with $N$ and $n$ denoting the number of cities for TSP and the number of qubits for binary-encoding a city index, respectively.  
    And it also needs $N$ qubits to temporarily store the checking results ${\rm OR}(\ket{I_k}),\, k=1,\cdots,N$. 
    In terms of the current scale of quantum computers, the algorithm requires a large number of qubits. 
    Below, we will focus on techniques to save qubits as much as possible to ensure that the qubit number requirements of our algorithm can be satisfied by existing quantum computers. 

\subsection{Improved Hamiltonian Cycle Detection Oracle.}
    {\bf Theorem~\ref{thm_cycle_determination}} provides the necessary and sufficient conditions of a Hamiltonian cycle, based on which we have designed the traveling and checking algorithms for the HCD oracle. 
    However, it is not necessary to track all the cities on the salesman's tour. 
    Instead, checking some of these tour steps suffices to decide a Hamiltonian cycle. This is based on our improvement of {\bf Theorem~\ref{thm_cycle_determination}}:
    \begin{theorem}[Improved Cycle Determination Theorem]  \label{thm_cycle_determination_2}
        Suppose the function $\pi: \{0,1,\cdots,N-1\}\to\{0,1,\cdots,N-1\}$ represents $\ket{C}$. Then $\ket{C}$ is a Hamiltonian cycle if and only if: 1) $\forall j|N: 1 \leqslant j \leqslant N-1$, $\pi^j(0) \neq 0$; \quad 2) $\pi^N(0) = 0$.
    \end{theorem}
    \begin{proof}
        The necessity is trivial. We prove the sufficiency of the theorem in the following. \par
        We claim that $\forall j: 1\leq j <N$, $\pi^j(0)\neq 0$. If not, then there exists some $1\leq j<N$ such that $\pi^j(0)=0$. We choose the minimum possible $j$. According to assumption (1), $j\nmid N$. Let $N=kj+r_0$, where $1\leq r_0<j$ is the remainder of $N$ divided by $j$. Thus, we have
        \begin{equation}
            0=\pi^N(0) = \pi^{kj+r_0}(0)=\pi^{r_0}(\pi^{kj}(0)) = \pi^{r_0}(0),
        \end{equation}
        which means that we find a smaller positive integer $r_0<j$ such that $\pi^{r_0}(0)=0$, leading to a contradiction. Thus, the sufficiency is proved immediately from {\bf Theorem~\ref{thm_cycle_determination}}. 
    \end{proof}
    Based on the above {\bf Theorem~\ref{thm_cycle_determination_2}}, we can replace the HCD checking step~\eqref{eqn:HCD_logical} with the following simplified formula 
    \begin{equation}
        R_{\rm HCD} = \left(\prod_{1\leq j <N, j|N} {\rm OR}(\ket{I_j}) \right)\oplus {\rm OR}(\ket{I_N}). 
    \end{equation} 
    This replacement significantly reduces the number of ${\rm OR}$ operations and the number of corresponding qubits that are required in the checking step of the HCD oracle. (see {\bf Fig.~\ref{fig:Hamiltonian Cycle Detector-improved}}) 

    \begin{remark}
        If the problem size $N$ is close to some large prime number, increasing the problem size may unexpectedly reduce the qubit consumption. 
        An interesting observation occurs when $N=16$. Since 16 has 5 divisors $\{1,2,4,8,16\}$, a total number of 5 checks are needed in the HCD oracle. However, if we increase the problem size by one city, i.e., $N=17$, only 2 checks are required, since 17 is a prime number. 
    \end{remark}

    \begin{figure}
        \centering
        \includegraphics[width=\linewidth]{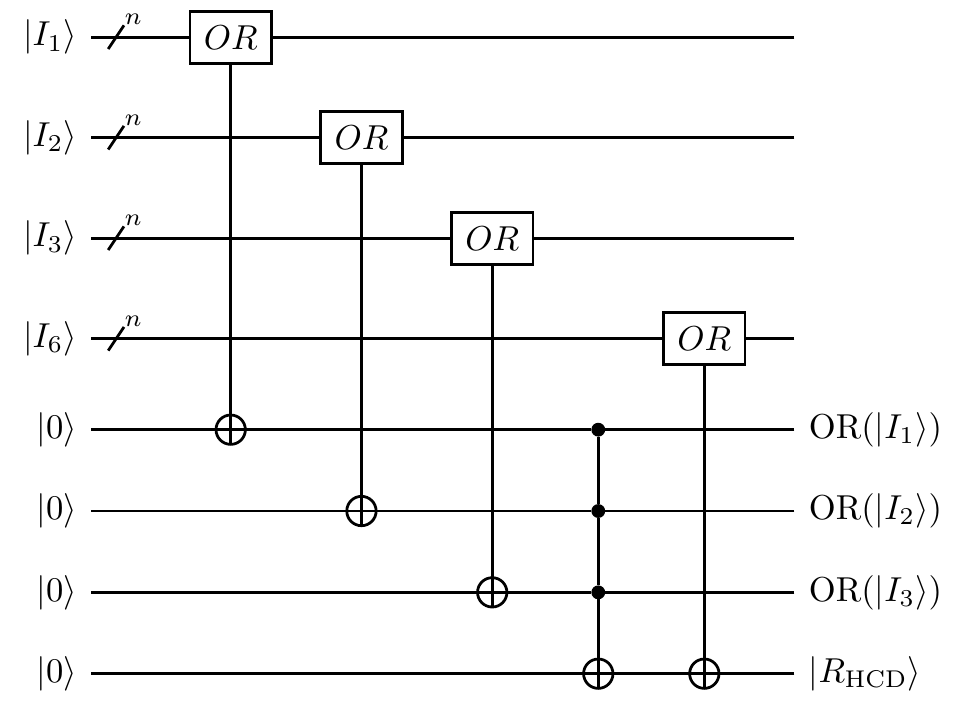}
            
        \caption{Quantum circuit illustration ($N=6$) of the checking part of the improved HCD oracle.}
        \label{fig:Hamiltonian Cycle Detector-improved}
    \end{figure}

\subsection{The anchor register strategy. }
    Qubit is always a scarce resource~\cite{almudever2017engineering} on both quantum simulators and real-world NISQ quantum computers.
    As we have already mentioned in the problem encoding scheme, $nN$ qubits are required to record the salesman's locations $\ket{I_j}$, which quickly becomes computationally unaffordable when applied to some non-trivial TSP problem sizes, for example, $N=6$, due to the scarcity of the qubit resource.  
    
    Fortunately, in the HCD oracle, we only need to store the result qubits ${\rm OR}(\ket{I_k})$ instead of the whole $n$-length location registers $\ket{I_k}$, implying that these location registers $\ket{I_k}$ may be reused.  
    On the other hand, directly erasing and reusing $\ket{I_k}$ after the calculation of $\ket{I_{k+1}}$ is prohibited since the erasure causes a quantum state collapse. 
    Thus, we introduce a small number of anchor registers $\ket{A_\ell}, \ell=1,2,\cdots,L$ to store part of the intermediate results. 
    These anchor registers enables us to reversibly uncompute most of the $\ket{I_k}$'s to save qubit usage. 
    If $L$ anchor registers and $k$ location registers are employed, the number of required qubits is reduced to $n(L+k)$, which significantly reduces the qubit consumption compared with storing all the $\ket{I_k}$. 
    This is realized by the following algorithm, where it is assumed that $N=L(k+1)$:

    \begin{enumerate}
        \item Initialization.
        \item For $i=1,2,\cdots,L$ do
        \begin{itemize}
            \item Apply $F$ gates to calculate the $k$ location registers $\ket{I_j}$, $(i-1)(k+1)+1\leq j \leq i(k+1)-1$ from $\ket{A_{i-1}}$.
            \item Calculate the corresponding ${\rm OR}(\ket{I_j})$.
            \item Calculate $\ket{A_{i}}$ by applying $F$ gate to $\ket{I_{i(k+1)-1}}$.
            \item Use $\ket{A_{i-1}}$ to sequentially free $\ket{I_j}$ by uncomputation. 
        \end{itemize}
        \item Return the checking results ${\rm OR}(\ket{I_j})$. 
    \end{enumerate}

    It is worth noting that, in the above algorithm, the assumption $N=L(k+1)$ can be dropped, i.e., $N$ need not be a multiple of $L$.  
    It suffices to let the number of anchor registers $L=\lfloor\frac{N}{k+1}\rfloor$, where $k$ is the number of intermediate location registers $\ket{I_{(i-1)(k+1)+1}},\cdots,\ket{I_{i(k+1)-1}}$. 
    For ${\rm OR}(\ket{I_j})$, this calculation is only needed when $j|N$. 
    For the uncomputation process, since the location registers are reused for different $i$, the total number of qubits required can be reduced from $nN$ to $n(k+L)$, at the cost of increasing the total circuit depth by a factor of 2 due to the uncomputation.  

    The optimal tradeoff between the number of anchor registers $L$ and the number of intermediate location registers $k$ are given by 
    \begin{equation*}
        k_{\rm opt}={\rm argmin}_k\left\{n\left(\lfloor\frac{N}{k+1}\rfloor+k\right)\right\},
    \end{equation*} 
    where the total number of required qubits is minimized. 

    The overall design principle is illustrated in {\bf Fig.~\ref{fig:simplified_oracle}}. 

    \begin{figure*}
        \centering
        \includegraphics[width=0.9\linewidth]{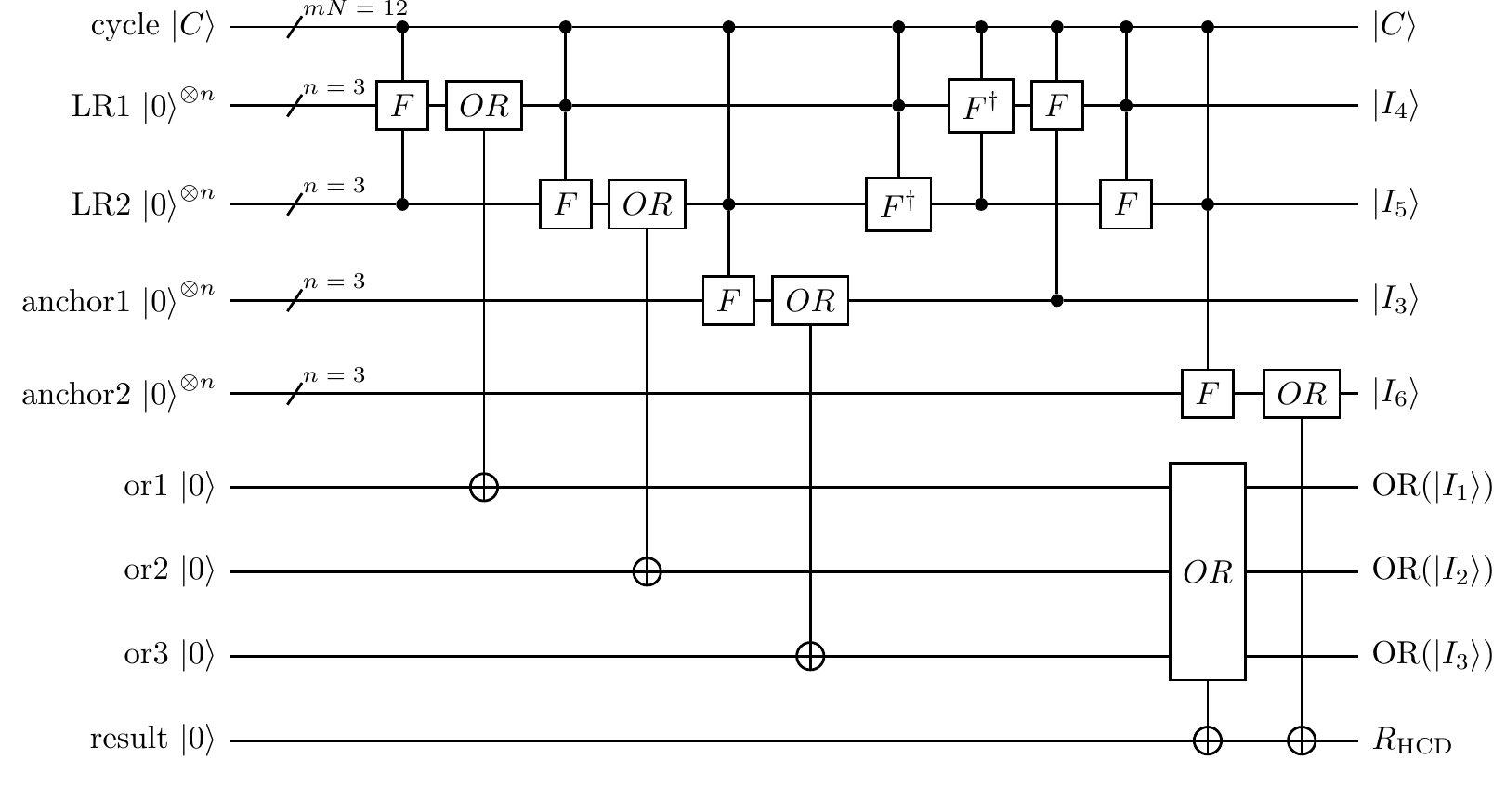}
    
        \caption{Quantum circuit of the improved HCD oracle with the ``anchor'' qubit-saving strategy. }
        \label{fig:simplified_oracle}
    \end{figure*}

\subsection{Simulation Results}
    In this section, we implement our proposed GQ-TSP for graphs with city number $N=4, 5, 6, 7$ and provide the numerical results. Specifically, we consider complete graphs for $N=4, 5$; and 4-sparse graph, i.e. graph degree $d=4$ for $N=6, 7$ due to the limitation of simulated qubit number.
    For $N\geq 5$, the sparse encoding is adopted to reduce the number of required qubits.  
    We generate our test graphs by uniformly picking $N$ cities within the square $[0,1]^2$, see {\bf Fig.~\ref{fig:results}(a-d)} for illustration. The distances between cities follow the Euclidean distance
    \begin{equation*}
        d_{ij}=\|x_i-x_j\|^2,\quad x_i, x_j\in[0,1]^2. 
    \end{equation*}
    
    The same quantum circuit is executed for $N_s=1024$ times to obtain statistical data of final quantum state $\ket{C}$. Then, the threshold $C_{\rm th}$ is updated based on the sampled results. For the CLC oracle, the QFT precision is set to $t=6$ qubits. 
    We implement our proposed GQ-TSP with the IBM Qiskit simulator, and run simulations with a Windows 11 PC equipped with an Intel i5-12400 CPU and 32GB RAM. 

    {\bf Fig.~\ref{fig:results}(e-h)} illustrate the sample probability of the top-3 valid cycles, where purple solid line, green dashed line and blue dotted line represent the 1st, 2nd and 3rd short cycle, respectively. 
    The sample probability curves exhibit a sinusoidal waveform, which coincides the rotation interpretation~\cite{nielsen2002quantum} of Grover's search algorithm. Specifically, the sample probability of the shortest cycle reaches the highest at the optimal Grover iteration $I_{\rm opt}$, which is close to the theoretical estimation $\hat{I}_{\rm opt}$ derived from the following formula~\cite{nielsen2002quantum} if the threshold $C_{\rm th}$ excludes all the sub-optimal solutions:
    \begin{equation}
        \hat{I}_{\rm opt} = {\rm CI}\left(\frac{\pi}{2\theta}-\frac{1}{2}\right),
        \label{eqn:iter_number}
    \end{equation}
    where $\theta=2\arcsin(1/2^{(m-1)/2})$, and ${\rm CI}(x)$ represents the integer closest to the real number $x$. The optimal sample probability of the top-3 cycles are shown in {\bf Table.~\ref{tab:sim}}, where $p_1, p_2, p_3$ represents the probability of 1st, 2nd, 3rd shortest cycle, respectively. The probability of the shortest cycle at $N=6$ is lower in that the length of the 2nd and 3rd short cycles are close to the shortest one. This also leads to $I_{\rm opt} < \hat{I}_{\rm opt}$ when $N=6$. We will look more deeply into this issue in the {\bf Discussion} section.

    \begin{table}
        \centering
        \begin{threeparttable}[b]
            \caption{Simulation Results} \label{tab:sim}
            \setlength{\tabcolsep}{3mm}
            \begin{tabular}{c c c c c c}
                \toprule[1.5pt]
                $N$ &  $I_{\rm opt}$ & $\hat{I}_{\rm opt}$ & $p_1$    & $p_2$    & $p_3$    \\ \hline
                4 & 8           & 8                 & 0.9873 & 0.0185 & 0.0136 \\
                5 & 17          & 17                & 0.9385 & 0.0234 & 0.0186 \\
                6 & 27          & 35                & 0.8695 & 0.0617 & 0.0590 \\
                7 & 62          & 71                & 0.8672 & 0.1250 & 0.0127 \\
                \bottomrule[1.5pt]
            \end{tabular}
        \end{threeparttable}
    \end{table}

    {\bf Fig.~\ref{fig:results}(i-l)} is the observed probability distribution over all the edges stored in the cycle register $\ket{C}$, when the number of iterations reaches $1/3$ of the optimal iteration number $I_{\rm opt}$. Darker edge color means the corresponding edge is observed at a higher probability. It can be seen from the figure that, the optimal solution can be observed at a reasonable probability even if only $1/3$ of the required iterations are performed. 
    {\bf Fig.~\ref{fig:results}(m-p)} shows the observed probability distribution over the edges when the optimal number of Grover iteration is reached. At this optimal number, the optimal cycle can be observed with the highest probability. As long as the optimal solution is unique, there exists such $C_{\rm th}$ that ensures the selection of the optimal solution when the QFT precision $t$ is sufficiently large. 
    
    By comparing along the columns of {\bf Fig.~\ref{fig:results}}, it is shown that the optimal TSP tour can be significantly amplified by the proposed GQ-TSP algorithm. 
    
    \begin{figure*}[t]
        \centering
        \includegraphics[width=0.8\linewidth]{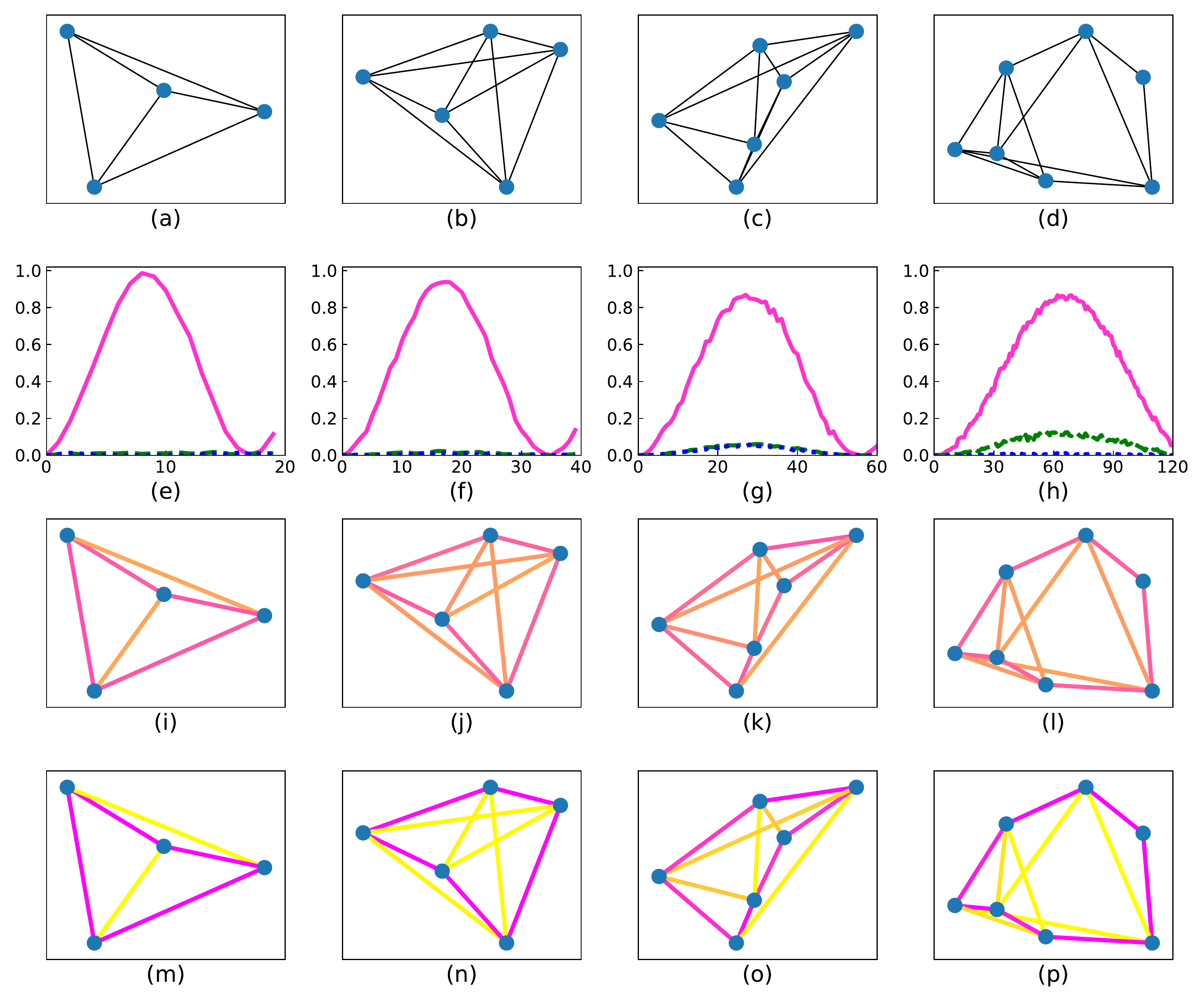}
        \caption{Numerical results of the proposed GQ-TSP. Columns 1-4 represent $N=4,5,6$, and $7$, respectively. }
        \label{fig:results}
    \end{figure*}

    In {\bf Table.~\ref{tab:perf comp}}, we provide the performance comparison between the proposed GQ-TSP algorithm and the VQA (including VQE and QAOA)-based baseline algorithms~\cite{khumalo2021investigation}. Criteria including success rate, qubit number and MT (mean running time on the Qiskit simulator) are considered in the comparison. The success rate (SR99) is the percentage amount of trials within 99\% of the optimal TSP solution. 

    As is shown in {\bf Table.~\ref{tab:perf comp}}, the proposed GQ-TSP has a significantly higher success rate (SR99) than VQE and QAOA. 
    Furthermore, as $N$ grows larger, the qubit consumption of the proposed GQ-TSP is also more efficient than the baselines. The consumption of the proposed algorithm only scales as $\mathcal{O}(N\log N)$ with a small constant factor guaranteed by the qubit-saving techniques, while both VQE and QAOA need exactly $N^2$ qubits. 

    In conclusion, compared with the baselines, the proposed GQ-TSP enjoys higher success rate and lower qubit consumption when $N$ is large. As a result, the GQ-TSP method will possibly be a promising TSP solver to run on a real-world fault-tolerant physical quantum computer.
    
    \begin{table*}
        \centering
        \begin{threeparttable}[b]
            \caption{Simulated Performance Comparison of Different Quantum TSP Algorithms}
            \label{tab:perf comp}
            \setlength{\tabcolsep}{2mm}
            \begin{tabular}{c|c|c|c|c|c|c|c|c|c}
                \toprule[1.5pt]
                
                \multirow{2}{*}{Size $N$}  &  \multicolumn{3}{c|}{VQE Method~\cite{khumalo2021investigation}}  & \multicolumn{3}{c|}{QAOA Method~\cite{khumalo2021investigation}} & \multicolumn{3}{c}{GQ-TSP\, (sparse/non-sparse)}  \\
                \cline{2-4} \cline{5-7} \cline{8-10}
                & SR99 & MT & \# qubits & SR99 & MT & \# qubits & SR99 & MT & \# qubits \\ \hline
                4  & 3.33  & 54.46s & 16 & 0 & 164.05s & 16 & 98.73   & 0.13h     & 23      \\
                5  & 26.67 & 92.17s & 25 & - & -       & 25 & 93.84   & 6.43h     & 25/31   \\
                6  & -     & -      & 36 & - & -       & 36 & 86.95   & 462.25h   & 31/39   \\
                7  & -     & -      & 49 & - & -       & 49 & 86.71   & 1340.5h   & 31/40   \\
                \bottomrule[1.5pt]
            \end{tabular}
            \begin{tablenotes}
                \item[$-$] The experimental data is not available due to either the absence of successful trials, or the difficulty of our classical simulators to handle quantum circuits of $\geq 31$ qubits. 
            \end{tablenotes}
        \end{threeparttable}
    \end{table*}

\section{Discussion}
    {\bf Qubit consumption.}
        Reducing qubit consumption is of practical importance, both from a simulating point of view, and from an NISQ-implementable aspect. 
        In our proposed GQ-TSP, the qubits are saved mainly by the qubit-efficient construction of the improved HCD oracle, the anchor register strategy, and the zeroed-ancilla reusing methods (will be thoroughly discussed in the {\bf Methods} section). 

        To gain quantitative insights into our proposed GQ-TSP algorithm, we analyze the total qubit usage in the case $N=6, m=2$, i.e., a graph containing 6 nodes with $d \leq 2^m=4$.
        First of all, according to the encoding method of $\ket{C}$, a number of $mN=2\times 6=12$ qubits are needed to encode a TSP cycle on a sparse graph, which also serve as the input of the CLC and HCD oracle. 
        Inside the HCD oracle, another set of quantum registers $\ket{I_k}$ are needed to store the locations of each tour step, consuming $n(\lfloor\frac{N}{k+1}\rfloor+k) = 12$ more qubits, with $n = \lceil \log_2N \rceil = 3$ being the qubit number required for natural encoding, $k = {\rm argmin}_{k'}\left(\lfloor\frac{N}{k'+1}\rfloor+k'\right) = 1$ being the number of $n$-qubit quantum registers in the HCD oracle, and $\lfloor\frac{N}{k+1}\rfloor = 3$ indicating the number of required anchor registers $\ket{A_\ell}$. 
        Besides, in the index forwarders ($F$ gates) of the oracle, $m=2$ additional qubits are used to temporarily store the result of the quantum multiplexer, which serves as input for the index converter (see {\bf Methods} for details). To perform checking operations of the tour, we employ $\sigma_0(N)-1 = 3$ qubits to store ${\rm OR}(\ket{I_1})$, ${\rm OR}(\ket{I_2})$, and ${\rm OR}(\ket{I_3})$, respectively. 
        Finally, two more qubits are used to store the results of the HCD oracle and the CLC oracle. 
        In conclusion, without considering the ancillary qubits for the CLC oracle, we consume a total number of $12+12+3+2+2=31$ qubits in the case of $N=6, m=2$ with sparse encoding.

        The CLC oracle is composed of the controlled $U$-operators, the QFT circuit, and the quantum comparator. The qubit consumption is the summation of all its components. 
        The QFT precision of the CLC oracle is set to $t=6$, which means that a number of $2t+1=13$ zeroed ancillas should be employed for the whole CLC oracle. 
        This is because a number of $t$ qubits are used to setup the uniform superposition state of $\ket{\ell}$, another $t$ qubits are used for the implementation of the controlled $U$-operators (see {\bf Methods} for details), and one ancillary qubit is used for storing the output of the CLC oracle.    
        However, since the HCD oracle consumes more ancillary qubits than the CLC oracle, we can reuse the ancilla qubits after the execution of the HCD oracle. 
        Thus, in the case of $N=6, m=2$, we still consume a total number of $12+12+3+2+2=31$ qubits, taking into consideration all the oracles and the qubit reusing strategies.
    
        Now we turn to the more general case. When the precision factor $t$ is fixed, the qubit consumption is directly determined by the problem size $N$ and the sparsity $m$ of the graph. 
        For the general case of sufficiently large problem size $N$, the qubit consumption is given by the following {\bf Theorem~\ref{thm:qubit usage}}.

        \begin{theorem} \label{thm:qubit usage}
            The qubit usage $n_{q}$ is $\mathcal{O}(mN+2\sqrt{N}\log_2N+\sigma_0(N))$, where $N$ is the number of cities, $m$ is the number of qubits needed to encode all neighbors of a city, and
            \begin{equation*}
                \sigma_{0} (N):=\#\{j\in\mathbb{Z}: 1\leq j \leq N, j|N\}.
            \end{equation*}
        \end{theorem}
        \begin{proof}
            The qubit usage of our proposed GQ-TSP is divided into three parts.
            \begin{enumerate}
                \item The cycle register $\ket{C}$ consumes $mN$ qubits. 
                \item A total number of $n(k+\lfloor\frac{N}{k+1}\rfloor)$ qubits are consumed inside the HCD and CLC oracles. This is because each city needs an $n$-qubit register to perform natural encoding, and $k+\lfloor\frac{N}{k+1}\rfloor$ copies of such register are required according to the anchor register strategy, where $k={\rm{argmin}}\left(\lfloor\frac{N}{k+1}\rfloor+k\right) \sim 2\sqrt{N},\, N\to\infty$ is the number of $n$-qubit quantum registers. 
                \item $\sigma_0(N)-1$ qubits are utilized to store the OR results ${\rm OR}(\ket{I_k})$. 
                \item Other qubit usages are all within $o(N)$, thus can be ignored asymptotically. As for the QFT circuits and the Quantum Comparator circuits, they do not consume extra qubits, since the ancilla qubits needed by them can be reused from the HCD oracle.
            \end{enumerate}
            Therefore, it suffices to introduce $\mathcal{O}(mN+2\sqrt{N}\log_2N+\sigma_0(N))$ qubits to execute this algorithm.
        \end{proof}
        According to {\bf Theorem~\ref{thm:qubit usage}}, a general approximation of the qubit usage is given by $\mathcal{O}(mN+2\sqrt{N}\log_2N+\sigma_0(N))$, which is asymptotically linear in $N$. 
        The table below shows the qubit consumption for different problem size $N$, illustrating the significant qubit-efficiency of the proposed GQ-TSP with all the optimization techniques, compared to the non-optimized naive implementation. 

        \begin{table}
            \centering
            \begin{threeparttable}
                \caption{Overall Qubit Consumption}
                \begin{tabular}{c|c|c|c|c|c}
                    \toprule[1.5pt]
                    \multirow{2}{*}{Size $N$} & \multirow{2}{*}{$n$} & \multirow{2}{*}{\# LR\tnote{$\dagger$}\;\;qubits} & \multicolumn{3}{c}{\# total qubits} \\ 
                    \cline{4-6} 
                        & & & 4-sparse & dense & non-optimized   \\
                    \hline
                    4  & 2 & 9    & 23   & 23   & 36      \\
                    5  & 3 & 12   & 25   & 31   & 46      \\
                    6  & 3 & 12   & 31   & 39   & 52      \\
                    7  & 3 & 20   & 31   & 40   & 58      \\
                    8  & 3 & 20   & 35   & 45   & 64      \\ \bottomrule[1.5pt]
                \end{tabular}
                \begin{tablenotes}
                    \item[$\dagger$] \# LR qubits represent the number of ancillary qubits in the Location Registers $\ket{I_k}$, which are used for storing the middle results of the HCD oracle. 
                \end{tablenotes}
            \end{threeparttable}
        \end{table}

        {\bf Circuit depth.}     
        The circuit depth is proportional to the running time of the quantum algorithm, which is both applicable on a classical simulator and on a real quantum hardware. 

        For the HCD oracles, the main building blocks are AND gates (C$^n$NOT gates), since most of the circuit depths are consumed by the index converters ($F$ gates), and the OR gates (realized by AND gates and $X$ gates).  The AND gate requires linear complexity on the number of its inputs. 
        As analyzed in the {\bf Methods} section, the QAQR of address length $n$ consumes $\mathcal{O}(Ln2^{n})$ depths. For the index forwarder, the address is of length $n=\lceil \log_{2}{N}\rceil$, and $L=m$, so each QAQR in the index forwarder consumes a depth of $\mathcal{O}(Nnm)$. 
        Another component of the index forwarder is the index converter (QACR), with $m + n$  input qubits and $n$ output qubits. It consumes a depth of $\mathcal{O}(n(m+n)2^{m+n})=\mathcal{O}(Ndn(m+n))$ (see {\bf Methods}). Thus, a single index forwarder consumes a depth of $\mathcal{O}(Nn((d+1)m+dn))$. During the computation of the HCD oracle, $N$ index forwarders are invoked. Therefore, the total depth for a HCD oracle is $\mathcal{O}(N^2 n((d+1)m+dn) )$, and if we assume $d$ be constant, the result is $\mathcal{O}(N^2 \log^{2}{N})$. 

        For the CLC oracles, the main part is the controlled $U$-operators. According to its construction, the depth of a single $U$-operator is $\mathcal{O}(2^{m})$. Since the CLC oracle requires one application of the $U$-operator to each of the $\ket{C_j},\,1\leq j\leq N$, the depth of the QPE module is $\mathcal{O}(N2^{m} + t^2)$, where $t$ is the number of precision qubits of the QFT subroutine. With the assumption that $d$ and $t$ are invariant with the increase of the problem size $N$, the total complexity of the CLC oracle is $\mathcal{O}(N)$.

        Finally, the GAS iteration should be repeated $\mathcal{O}(2^{Nm /2})$ times according to \eqref{eqn:iter_number} to ensure optimal amplitude amplification. 
        To sum up, the overall depth of the whole GQ-TSP algorithm is $\mathcal{O}(2^{mN /2}N^2 \log^2{N})$.

\appendix

\section*{Methods}

\subsection{Implementation of the CLC oracle.}
    The CLC oracle contains three main parts: the $U$-operators, the inverse quantum Fourier transform (iQFT) module, and the quantum comparator. The iQFT module~\cite{nielsen2002quantum} is well-studied, so we implement it with the textbook techniques. 
    The quantum comparator with a classical pre-defined threshold $C_{\rm sh}$ can be implemented directly in Qiskit. Consequently, the difficulty of realizing this CLC oracle mainly lies in the implementation of the $U$-operators, which we will discuss in detail. 

    {\it Pre-processing.} Since the iQFT is subject to a $2\pi$ phase ambiguity, the maximal allowed cycle length should be normalized to be smaller than $2\pi$. 
    Note that each term of the total cycle length, i.e., the components of a $U$-operator $U_{j}:\, 1\leq j\leq N$, are determined by the parameters $\theta_{j,k}$. 
    Thus, in order to avoid such ambiguity, we normalize the adjacency matrix $A$ by the sum of all of its entries, and then let: $\theta_{j,k} = 2\pi a_{k,P_{j}[k]}$ for $0 \leq j < N$ and all valid $k$. 

    Since $U_{j}$ acts on the $m$-qubit register $\ket{C_j}$, the matrix representation of $U_j$ has both $2^{m}$ rows and columns. 
    Hoever, $\#P_j$ may be strictly less than $2^m$, leading to the inability to fill the diagonal of a unitary matrix $U_j$. 
    To address this problem, these unspecified diagonal entries are filled with 1's ( $\exp(\ri 0)$ ), and accordingly, the optimization target of the GAS-QTSP is shifted to maximizing the cycle length instead of minimizing it, in order to exclude these unspecified paths with zero costs. 
    Thus, the problem pre-processing shoud start with converting the cost-minimization problem in to an equivalent maximization one.   

    {\it Construction.} The controlled $U$-operators are constructed recursively in our experiments. For the simplest case $m=1$, since the $U$-operator needs to provide $2^m=2$ different phase shifts, it is just a controlled-phase gate:
    \begin{equation}
        \begin{aligned}
            U = \begin{bmatrix} e^{\ri \theta_0} & 0 \\ 0 & e^{\ri\theta_1} \end{bmatrix} = e^{\ri\theta_0}
            \begin{bmatrix} 1 & 0 \\ 0 & e^{\ri(\theta_1-\theta_0)} \end{bmatrix},
        \end{aligned}
    \end{equation}
    where the controlled-$U$ gate can be represented by 
    \begin{equation}
        \text{controlled-}U = 
        \begin{bmatrix} 1 & & & \\
            & 1 & & \\
            & & e^{\ri \theta_0} &  \\ 
            & &  & e^{\ri \theta_1} 
        \end{bmatrix}.
    \end{equation}
    For $m=2$, we design a special quantum circuit to implement the unitary operator 
    \begin{equation*} 
        U = {\rm diag}\left( \exp\left(\ri\theta _{0}\right),  \exp\left(\ri\theta _{1}\right), \exp\left(\ri\theta_{2}\right), \exp\left(\ri\theta _{3}\right)\right).
    \end{equation*}
    Let 
    \begin{equation*}
        x=-\frac{1}{2}\left( \theta_3 - \theta_2 - \theta_1 + \theta_0 \right),
    \end{equation*} 
    then, inspired by the paper~\cite{srinivasan2018efficient}, the corresponding quantum circuit implementing $U$ can be decomposed into two operators $V_1$ and $V_2$:
    \begin{subequations}
        \begin{align}
            V_1 = &  \begin{bmatrix} e^{\ri\theta_0} & 0 & 0 & 0 \\ 0 & e^{\ri\theta_1} & 0 & 0 \\ 0 & 0 & e^{\ri\theta_2} & 0 \\ 0 & 0 & 0 & e^{\ri(\theta_2 + \theta_1 - \theta_0)} \end{bmatrix} \\
            V_2 = &  \begin{bmatrix} 1 & 0 & 0 & 0 \\ 0 & 1 & 0 & 0 \\ 0 & 0 & 1 & 0 \\ 0 & 0 & 0 & e^{-\ri 2x} \end{bmatrix},
        \end{align}
    \end{subequations}
    where one can easily verify that $U=V_1 V_2$. The controlled-$V_1$ operator is implemented in {\bf Fig.~8} by further decomposing it into controlled phase gates. The implementation of controlled-$V_2$ gate is described in {\bf Fig.~9}. In both figures, $P(\theta)$ denotes the controlled phase gate with rotation angle $\theta$.

    \begin{figure}
        \centering
        \includegraphics[width=0.85\linewidth]{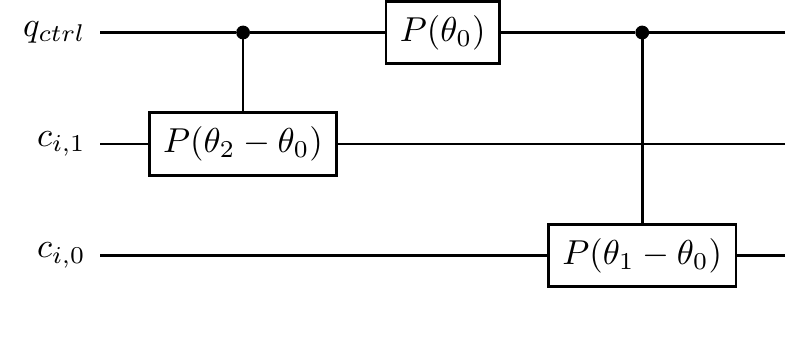} 
        
        \caption{Implementation of operator $V_1$, where $m=2$. $q_{\rm ctrl}$ is the control qubit; $c_{i,1}$ and $c_{i,0}$ are qubits in $\ket{C_i}$; $\theta_0, \theta_1, \theta_2, \theta_3$ correspond to phase shifts on states $\ket{100}, \ket{101}, \ket{110},$ and $\ket{111}$, respectively. }
        \label{n1_2_circ_ex1}
    \end{figure}
        
    \begin{figure}
        \centering
        \includegraphics[width=1\linewidth]{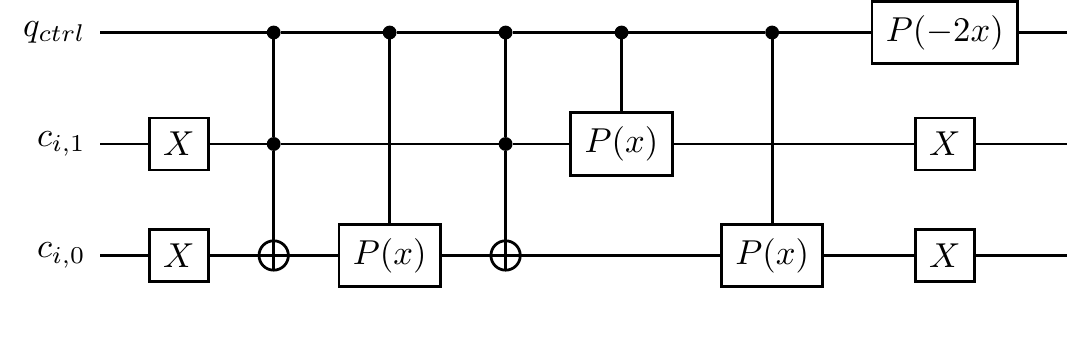} 

        \caption{Implementation of operator $V_2$, where $m=2$. }
        \label{fig:n1_2_circ_ex2}
    \end{figure}

    For larger $m>2$, we use recursive construction by only one extra ancillary qubit. Since there are $m$ ``choice-encoding'' qubits $\ket{c_{j,(m-1)} c_{j,(m-2)}\cdots c_{j,0}}$, we can partition matrix $U_j$ along the diagonal into two smaller matrices, with each of the sub-matrix inheriting half of the original values $\theta_{j,k}$. Thus, we can recursively construct two smaller $U$-operators, controlled by the first ``choice-encoding'' qubit $\ket{c_{j,(m-1)}}$ and its flipped version $X\ket{c_{j,(m-1)}}$, respectively. Note that an ancillary qubit is introduced to temporarily save the AND result of the control qubit $q_{\rm ctrl}$ and the first qubit $\ket{c_{j,(m-1)}}$. 

\subsection{Efficient implementation of quantum logical gates.}
    The implementation of AND gate, i.e. CNOT with $n$ control qubits (C$^n$NOT), is a basic component of our TSP quantum circuit, and then OR gates can be easily constructed from C$^n$NOT and X gates. Naive realization of C$^n$NOT consumes $(n-2)$ additional ancillas and $(2n-3)$ Toffoli gates. Though it is a linear consumption on qubits and circuit depth, it is still possible to reduce the number of ancillas to 1, while keeping the number of Toffoli gates within $\mathcal{O}(n)$. In state-of-the-art quantum technology, qubits are   computational resources that are regarded to be much more expensive than circuit depth, which is the same case in classically simulated quantum circuits. So our main efforts are devoted to qubit-reducing.

    However, is it possible to implement the C$^n$NOT gate without any ancillas? In fact, it is impossible to do the C$^n$NOT operation only through basic Toffoli gates. The reason is that, the C$^n$NOT gate swaps the all-one state $\ket{11\cdots 11}$ and the state $\ket{11\cdots 10}$, and leaves other states unchanged, so it is an odd permutation on $2^n$ elements. However, the basic Toffoli gates do not touch all the qubits, thus being even permutations. Cascading even permutations cannot result in an odd permutation, so we cannot implement C$^n$NOT only by basic gates only without an ancilla. 
    
    Our $O(n)$ construction requires only 1 ``zeroed'' ancilla qubit whose initial state is $\ket{0}$, which is proposed in the paper~\cite{xu2015reversible}. To fulfill this target, we first decompose the C$^n$NOT into four C$^{n/2}$NOT gates (approximately $n/2$ if $n$ is odd), and then implement each half-sized gate with the ``borrowed'' ancilla technique. Thus, the construction of universal C$^n$NOT consuming only 1 borrowed qubit $b$ can be achieved within only 4 steps, which is shown in the following
    \begin{enumerate}[(1)]
        \item Toggle $b$ conditioned on $q_{0:\lfloor n/2\rfloor}$. Use $q_{(\lfloor n/2\rfloor +1):(n-1)}$ as borrowed ancillas,
        \item Toggle $r$ conditioned on $q_{(\lfloor n/2\rfloor +1):(n-1)}$ and $b$. Use $q_{0:\lfloor n/2\rfloor}$ as borrowed ancillas,
        \item Toggle $b$ conditioned on $q_{0:\lfloor n/2\rfloor}$. Use $q_{(\lfloor n/2\rfloor +1):(n-1)}$ as borrowed ancillas,
        \item Toggle $r$ conditioned on $q_{(\lfloor n/2\rfloor +1):(n-1)}$ and $b$. Use $q_{0:\lfloor n/2\rfloor}$ as borrowed ancillas,
    \end{enumerate}
    where $r$ denotes the result qubit, and $q_{0:(n-1)}$ denotes the $n$ input qubits. Step (2) and step (4) together form a toggle-detection circuit on $b$, whose toggling is conditioned on $q_{0:\lfloor n/2\rfloor}$. Step (1) followed by (3) ensures the borrowed ancilla $b$ to be unaffected. Thus, we convert the construction problem of C$^n$NOT into two C$^{\lfloor n/2\rfloor+1}$NOTs and two C$^{n-\lfloor n/2\rfloor}$NOTs, and then implement the four CNOTs by borrowing qubits from each other. Assume $n$ is even, then the circuit depth is $8(n-3)$ in total (counted in Toffolis), which takes linear time to execute within constant number of ancillas.

    \begin{figure}\label{c4not_borrowed_ancilla}
        \centering
        \includegraphics[width=0.9\linewidth]{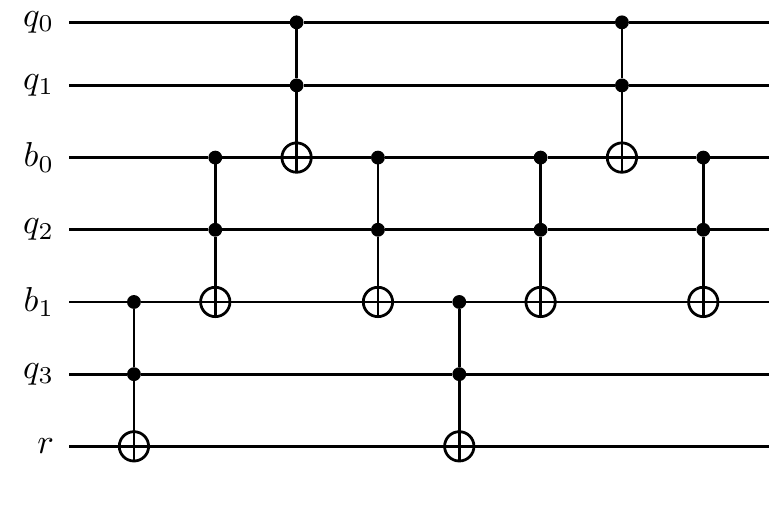} 
            
        \caption{C$^4$NOT implemented with 2 borrowed ancillas $b_0$ and $b_1$. $q_{0:3}$ are the input qubits; $r_0$ is the result qubit. Note that the result qubit $r$ is flipped if and only if $q_3 =\ket{1}$, and $b_1$ is flipped by the two Toffoli gates that act on $b_1$. And also, $b_1$ is flipped if and only if $q_2 = \ket{1}$, and $b_0$ is flipped by the left-top Toffoli. Since each flipping is conditioned both on an input qubit being $\ket{1}$ and on the previous flipping, it finally constitutes a C$^4$NOT on $r$. The right half side of the circuit undoes the flipping operations on borrowed qubits and restores them to their initial states, ensuring the borrowed ones to be unaltered. The number of borrowed ancillas needed is $(n-2)$ and the circuit depth is $4(n-2)$, respectively.}
    \end{figure}

\subsection{Realization of the index forwarder. }
    The index forwarder $F$ plays an important role in the traveling part of the HCD oracle, since it forward-computes the next city index $\ket{I_{k+1}}$ given the current index $\ket{I_k}$.  
    An index forwarder in each step consists of two parts: a {\it quantum multiplexer} to compute the sparse-encoded index of the next city $\ket{S_{k+1}}$, and an {\it index converter} to get the index of the next city $\ket{I_{k+1}}$ from $\ket{S_{k+1}}$. 
    From a general point of view, both the {\it quantum multiplexer} and the {\it index converter} are instances of the quantum-addressed register.
    However, the difference is that the quantum multiplexer  extracts quantum data from the input cycle register $\ket{C}$, while the index converter extracts classical data in a fixed classical lookup table. 
    In practice, the {\it quantum multiplexer} is implemented by quantum-addressed quantum registers (QAQR), while the {\it index converter} is implemented by quantum-addressed classical registers (QACR). The construction of the QAQR and QACR from basic quantum gates will be elaborated in the following, but before that we need to introduce what are registers. 

    Registers are essential memory components in modern classical processors. 
    Synchronized by a clock, classical registers can take in the input bits at each clock rise, and keep its value unchanged throughout the clock period. 
    Mathematically, a register stores a numerical value $R$. Several registers form a register file (RF), which can be represented by a group of $2^n$ length-$L$ bit arrays $R[j]\in\{0,1\}^L$, where $0\leq j<2^n$ is the $n$-bit index, or the ``address'' in computer science. 
    Given a certain address $j$, the value $R[j]$ stored in the RF can be immediately accessed. 
        
    In classical registers, the address $j$ is a stable deterministic value during the clock period. However, different from classical registers, a quantum-addressed quantum register (QAQR) allows $j$ to be in superposition, and it outputs all the possible $R[j]$ simultaneously, following the original superposition coefficients. 
    Since quantum data are stored in qubits, the function of a QAQR is to select a certain group of qubits (the data) conditioned on another group of qubits (the address). 
    To realize this function, we design the quantum multiplexers. 
    A single-qubit $n$-addressed quantum multiplexer is addressed by $n$ qubits $A_{n-1}A_{n-2}\cdots A_{0}$, and the output of the QAQR can be expressed by Boolean algebra as follows (``$\oplus$'' denotes modulo-2 addition, i.e. XOR)
    \begin{equation}
        S=\bigoplus_{j=0}^{2^n-1}\left(\prod_{k=0}^{n-1}A_k^{(j_k)}\right)R[j],
        \label{Qreg_expansion}
    \end{equation}
    where $(j_{n-1}j_{n-2}\cdots j_0)_2$ is the binary representation of $j$, and 
    \begin{equation*}
        A_k^{(b)} = \left\{ \begin{aligned}  
            {\rm not}({A}_k), \quad &b=0\\
            A_k,\quad & b=1.
        \end{aligned}\right.
    \end{equation*}
    The QAQR is then defined by 
    \begin{equation}
        {\rm QAQR}\ket{A}\ket{0} = \ket{A}\ket{S}.
    \end{equation}

    Since in~\eqref{Qreg_expansion} there are $2^n$ terms, where each term includes $n$ AND operations, direct implementation~\eqref{Qreg_expansion} of this single-qubit $n$-addressed multiplexer requires $2^n$ copies of $(n+1)$-input AND gates and $n2^{n-1}$ X gates. The AND gate consumption is reasonable because the scale of input data $R[j]$ is also $\mathcal{O}(2^n)$, but we can reduce the number of X gates to $2^n$, i.e., a linear reduction, by traversing all the terms following the Gray code~\cite{bitner1976efficient}. Thus, we have designed a quantum-addressed multiplexer that consumes $\mathcal{O}(2^n)$, i.e., linear number of gates with respect to the scale of the input data. Based on this design of single-qubit $n$-addressed quantum multiplexer, we can contruct an $n$-addressed quantum register file of length $L$ qubits, by simply stacking $L$ copies of single-qubit $n$-addressed quantum multiplexer, within $\mathcal{O}(Ln2^n)$ gates (the factor $n$ comes from implementing the AND gate). 

\subsection{Algorithm Design}
    Our GQ-TSP algorithm can be summarized as follows: 

    \begin{breakablealgorithm}
        \caption{GQ-TSP Algorithm} \label{alg:QTSP}
        \newcommand{\nl}{\Statex}
        \setstretch{1.35}
        \begin{algorithmic}[1] 
            \Require Quantum registers of lengths $mN, t, 1, 1, 1$ qubits, represented by $\ket{C}$, $\ket{T}$, $\ket{R_{\rm CLC}}$, $\ket{R_{\rm HCD}}$, and $\ket{R}$, respectively. 
            \Ensure The cycle with the lowest cost $C^*$
            \State $i \leftarrow 0$.
            \State $I_{\rm opt} \leftarrow \left\lceil\pi \sqrt{2^{mN}} / 4\right\rceil$.
            \State $\ket{C} \leftarrow \ket{0}^{\otimes mN}$, $\ket{T} \leftarrow\ket{0}^{\otimes t}$.
            \State $\ket{R_{\rm CLC}}\leftarrow\ket{0}$,  $\ket{R_{\rm HCD}}\leftarrow\ket{0}$.
            \State $\ket{R}\leftarrow\ket{-}$.
            \State Select a proper comparing threshold $C_{\rm th}$ from several samples on the graph.
            \State Initialize $\ket{C}$ and $\ket{T}$ to uniform superposition state: \nl
            $\ket{C}\otimes\ket{T}\otimes\ket{R_{\rm CLC}}\otimes\ket{R_{\rm HCD}}\otimes\ket{R}= \sum_C a_C\ket{C}\otimes \sum_k\ket{k} \otimes \ket{0}\otimes\ket{0}\otimes\ket{-}$ ($a_C=\frac{1}{\sqrt{2^{mN}}}$)

            \While{ $i< I_{\rm opt}$}
                \State Apply controlled-$U$ gates to $\ket{T}$, controlled by $\ket{C}$, $c={\rm cost}(\ket{C})$: \nl $\rightarrow \sum_C \left(a_C\ket{C}\otimes \sum_k e^{2\pi\ri kc/2^t}\ket{k}\otimes \ket{0}\otimes\ket{0}\otimes\ket{-}\right)$

                \State Apply iQFT to $\ket{T}$, obtaining the cost of each $\ket{C}$: \nl $\rightarrow \sum_C \left(a_C\ket{C}\otimes \ket{{c}} \otimes \ket{0}\otimes\ket{0}\otimes\ket{-}\right)$

                \State Use quantum comparator to compare $\ket{c}$ with the threshold $C_{\rm th}$: \nl $\rightarrow \sum_C \left(a_C\ket{C}\otimes \ket{{c}} \otimes \ket{\mathbbm{1}_{c<C_{\rm th}}}\otimes\ket{0}\otimes\ket{-}\right)$

                \State Uncompute and restore $\ket{T}$ to the initial state \nl$\rightarrow \sum_C \left(a_C\ket{C}\otimes \sum_k \ket{k} \otimes \ket{\mathbbm{1}_{\{c<C_{\rm th}\}}}\otimes\ket{0}\otimes\ket{-}\right)$

                \State Use HCD oracle to set $\ket{R_{\rm HCD}}$: \nl $\rightarrow \sum_C \left(a_C\ket{C}\otimes \sum_k \ket{k} \otimes \ket{\mathbbm{1}_{\{c<C_{\rm th}\}}}\otimes\ket{\mathbbm{1}_{\{C~{\rm valid}\}}}  \right.$\nl$\left.\otimes\ket{-}\right)$

                \State Apply Toffoli gate on $\ket{R}$, where $\mathbbm{1}_C = \mathbbm{1}_{\{c<C_{\rm th}\}}\cdot\mathbbm{1}_{\{C~{\rm valid}\}}$: \nl $\rightarrow \sum_C \left(a_C\ket{C}\otimes \sum_k \ket{k} \otimes \ket{\mathbbm{1}_{\{c<C_{\rm th}\}}}\otimes\ket{\mathbbm{1}_{\{C~{\rm valid}\}}} \right.$ \nl $\left.\otimes (-1)^{\mathbbm{1}_{C}}\ket{-}\right)$

                \State Uncompute $\ket{R_{\rm HCD}}$: \nl $\rightarrow \sum_C \left((-1)^{\mathbbm{1}_{C}}a_C \ket{C}\otimes \sum_k \ket{k} \otimes \ket{\mathbbm{1}_{\{c<C_{\rm th}\}}} \right.$\nl$\left.\otimes\ket{0} \otimes \ket{-}\right)$

                \State Uncompute $\ket{R_{\rm CLC}}$: \nl $\rightarrow \sum_C \left((-1)^{\mathbbm{1}_{C}}a_C \ket{C}\otimes \sum_k \ket{k}\otimes \ket{0}\otimes\ket{0} \otimes \ket{-}\right)$

                \State Apply the diffusion operator, which is equivalent to updating the coefficient $a_C$: \nl$\rightarrow (2\ket{\psi}\bra{\psi}-I)\sum_C \left((-1)^{\mathbbm{1}_{C}}a_C \ket{C}\otimes \sum_k \ket{k}\otimes \ket{0} \right.$\nl$\left.\otimes\ket{0} \otimes \ket{-}\right)$

                \State $i \leftarrow i+1$
            \EndWhile
            \State $C^* \leftarrow$ Measurement result of $\ket{C}$
        \end{algorithmic}
    \end{breakablealgorithm}

    It is worth noting that, more ancillary qubits may be used during the execution of CLC and HCD oracles, which are not directly shown in the input of {\bf Algorithm~\ref{alg:QTSP}}.


\bibliographystyle{naturemag}
\bibliography{nature_template}

\pagebreak
\newpage
\onecolumngrid
\beginsupplement

\section*{Supplementary Numerical Results.}
    We present our experimental results in the simplest cases: $N=4,5,6$. To ensure that our numerical simulation experiment can be carried out by a classical computer within reasonable time, we assume that the maximum degree of the TSP graph does not exceed $4$, i.e. $d\leq 4$. So $m =2$ is sufficient for encoding each path choice. 

    \begin{figure}
        \centering
        \begin{minipage}{1\linewidth}
            \centering
            \includegraphics[width=1\linewidth]{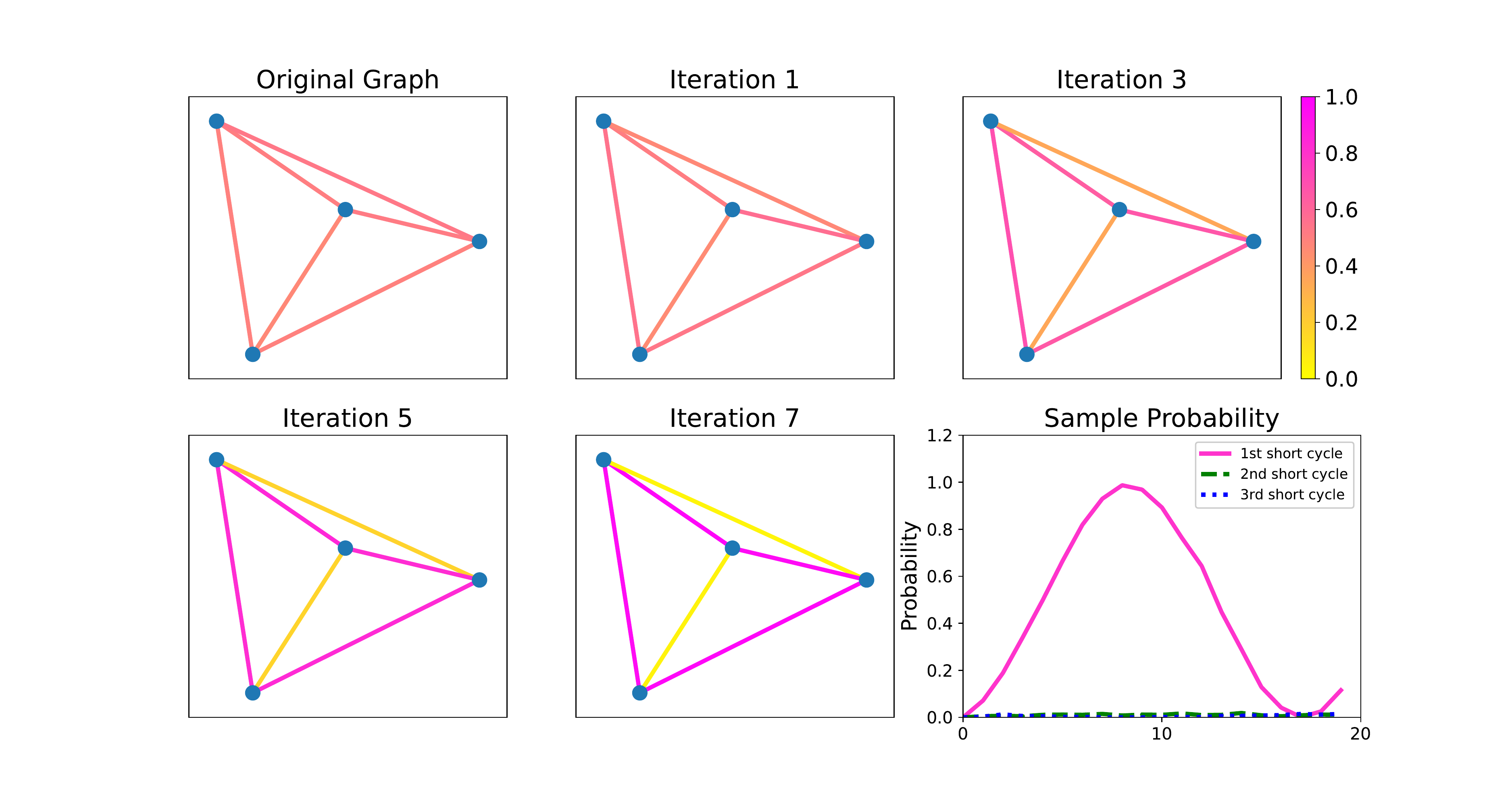}
        \end{minipage}
        \caption{$N=4$, statistical accumulative data of the output of our quantum circuit.}
        \label{fig:N4_exp_results}
    \end{figure}
        
    \begin{figure}
        \centering
        \begin{minipage}[t]{1\linewidth}
            \centering 
            \includegraphics[width=1\linewidth]{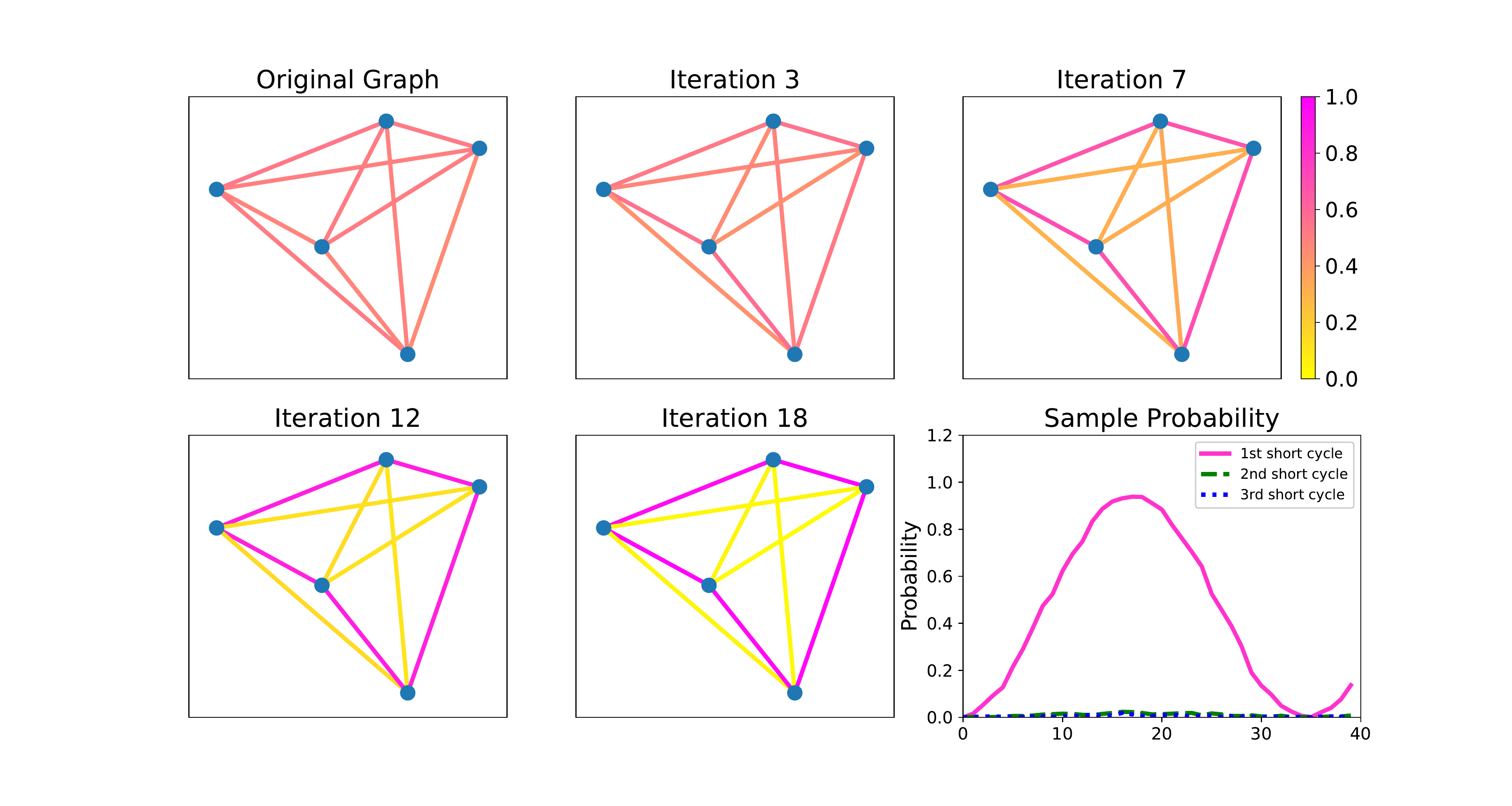}
        \end{minipage}
        
        \caption{$N=5$, statistical accumulative data of the output of our quantum circuit.}
        \label{fig:N5_exp_results}
    \end{figure}

    The graphs are randomly generated with edge weights drawn from a uniform distribution on $[0,1]$, and subject to the degree constraint $d_v\leq 4$. After different numbers of iterations, we sampled the output of the quantum circuit, and recorded the frequency of each edge being observed. As we can see in the first five subfigures of each case, the frequency of the shortest path grows steadily as the number of Grover iterations increases. Furthermore, the sample probability of the shortest path also takes a sinusoidal waveform, which is consistent with the theory of a Grover rotation. 

    For each case $N=4,5,6,7$, we can empirically find the optimal iteration number $I_{\rm opt} = \argmax_kP_1(k)$ where the probability of observing the shortest cycle reaches its maximum, where $P_1(k)$ denotes the observing probability of the 1-st short cycle after $k$ Grover iterations. 

    \begin{table}[H]
        \centering
        \begin{tabular}{p{3cm}<{\centering} | p{3cm}<{\centering} | p{3cm}<{\centering}}
        \toprule[1.5pt]
            $n$ & $N$   & $I_{\rm opt}$ \\ \hline
            4 & 16  & 8         \\  
            5 & 32  & 18        \\  
            6 & 64  & 27        \\  
            7 & 128 & 70        \\  \bottomrule[1.5pt]
        \end{tabular}
    \end{table}

\end{document}